\documentclass[runningheads]{llncs}
\usepackage{graphicx}
\newcommand{\ie}{\textit{i.e.}\xspace}
\newcommand{\etal}{\textit{et al.}\xspace}
\newcommand{\wrt}{\textit{w.r.t.}\xspace}
\newcommand{\eg}{\textit{e.g.}\xspace}
\newcommand{\etc}{\textit{etc.}\xspace}

\newcommand{\mci}{\textsc{Minimum Connectivity Inference}\xspace}
\newcommand{\mcishort}{\textsc{MCI}\xspace}
\newcommand{\cga}{\textsc{CGA}\xspace}
\newcommand{\ilp}{\textsc{Flow-MILP}\xspace}

\newcommand{\E}{\mathcal{E}}
\newcommand{\B}{\mathcal{B}}
\newcommand{\C}{\mathcal{C}}
\renewcommand{\P}{\mathcal{P}}
\newcommand{\M}{\mathcal{M}}

\usepackage[linesnumbered,ruled,vlined]{algorithm2e}
\SetKwInput{KwData}{Input}
\SetKwInput{KwResult}{Output}

\usepackage{array}
\usepackage{float}
\newcolumntype{C}[1]{>{\centering\arraybackslash}m{#1}}

\usepackage[utf8]{inputenc}

\begin{document}
\title{Constraint Generation Algorithm for the Minimum Connectivity Inference Problem}
\titlerunning{Constraint Generation Algorithm for the Minimum Connectivity Problem}
%
\author{\'Edouard Bonnet\inst{1} \and
Diana-Elena Fălămaş\inst{1,2} \and
R\'emi Watrigant\inst{1}}
\authorrunning{\'E. Bonnet, D. Fălămaş, and R. Watrigant}
%
\institute{Univ Lyon, CNRS, ENS de Lyon, Universit\'e Claude Bernard Lyon 1,
LIP, F-69342, LYON Cedex 07, France \and
Technical University of Cluj-Napoca, Romania.  \\
~\\
 \email{\{edouard.bonnet, remi.watrigant\}@ens-lyon.fr}, \email{falamasd@yahoo.com}  }
\maketitle              
\begin{abstract}

Given a hypergraph $H$, the \mci problem asks for a graph on the same vertex set as $H$ with the minimum number of edges such that the subgraph induced by every hyperedge of $H$ is connected.
This problem has received a lot of attention these recent years, both from a theoretical and practical perspective, leading to several implemented approximation, greedy and heuristic algorithms. Concerning exact algorithms, only Mixed Integer Linear Programming (MILP) formulations have been experimented, all representing connectivity constraints by the means of graph flows. 
In this work, we investigate the efficiency of a \textit{constraint generation algorithm}, where we iteratively add cut constraints to a simple ILP until a feasible (and optimal) solution is found.
It turns out that our method is faster than the previous best flow-based MILP algorithm on random generated instances, which suggests that a constraint generation approach might be also useful for other optimization problems dealing with connectivity constraints.
At last, we present the results of an enumeration algorithm for the problem.

\end{abstract}

\section{Introduction and related work}\label{sec:intro}
We study the problem where one wants to infer a binary relation over a set of items $V$ (that is, a graph), where the input consists of some subsets of those items which are known to be connected in the solution we are looking for. In other words, the input can be represented by a hypergraph $H=(V, \E)$, and we are looking for an underlying undirected graph $G=(V, E)$ such that for every hyperedge $S \in \E$, the subgraph induced by $S$, denoted by $G[S]$, is connected (such a graph $G$ will be called a \textit{feasible solution} in the sequel). 
 Observe that it is easy to construct trivial feasible solutions to this problem: consider for instance the graph $K(H)$ having vertex set $V$ and an edge $uv$ iff $u$ and $v$ belong to a same hyperedge. Since these solutions are unlikely to be of great interest in practice, it makes sense to add an optimization criteria. In this paper, we focus on minimizing the number of edges of the solution. More formally, we study the following problem:\\

\begin{minipage}{\textwidth}
\mci (\mcishort)\\
\underline{Input:} a hypergraph $H=(V, \E)$\\
\underline{Output:} a graph $G=(V, E)$ such that $G[S]$ is connected $\forall S \in \E$\\
\underline{Goal:} minimize $|E(G)|$\\
\end{minipage}

This optimization problem is NP-hard \cite{DuMi88}, and was first introduced for the design of vacuum systems \cite{DuMi95}. It has then be studied independently in several different contexts, mainly dealing with network design: computer networks \cite{FaHuWuEr08}, social networks \cite{AnAsReHuVoZe10} (more precisely modeling the \textit{publish/subscribe} communication paradigm \cite{ChMeToVi07,HoHrIzOnStWa12,OnRi09}), but also other fields, such as auction systems \cite{CoDeSa04} and structural biology \cite{AgArCaCaCo13,AgArCaCaCo15}. Finally, we can mention the issue of hypergraph drawing, where, in addition to the connectivity constraints, one usually looks for graphs with additional properties (\eg planarity, having a tree-like structure... \etc) \cite{BrCoPaSa12,JoPo87,KlMcNo14,KoSt03}. This plethora of applications explains why this problem is known under different names, such as \textsc{Subset Interconnected Design}, \textsc{Minimum Topic Overlay} or \textsc{Interconnection Design}. For a comprehensive survey of the theoretical work done on this problem, see \cite{ChKoNiSoSuWe15} and the references therein.

Concerning the implementation of algorithms, previous works mainly focused on approximation, greedy and other heuristic techniques \cite{OnRi09}.
To the best of our knowledge, the first exact algorithm was designed by Agarwal \etal \cite{AgArCaCaCo13,AgArCaCaCo15} in the context of structural biology, where the sought graph represents the contact relations between proteins of a macro-molecule, which has to be inferred from a hypergraph constructed by chemical experiments and mass spectrometry. In this work, the authors define a Mixed Integer Linear Programming (MILP) formulation of the problem, representing the connectivity constraints by flows. They also provide an enumeration method using their algorithm as a black box, by iteratively adding constraints to the MILP in order to forbid already found solutions. Both their optimization and enumeration algorithms were tested on some real-life (from a structural biology perspective) instances for which the contact graph was already known.

This MILP model was then improved recently by Dar \etal \cite{DaFiMaSc18}, who mainly reduced the number of variables and constraints of the formulation, but still representing the connectivity constraints by the means of flows. In addition, they also presented and implemented a number of (already known and new) reduction rules.
This new MILP formulation together with the reduction rules were then compared to the algorithm of Agarwal \etal on randomly-generated instances. 
For every kind of tested hypergraphs (different number and sizes of hyperedges), they observed a drastic improvement of both the execution time and the maximum size of instances that could be solved.

In this paper we initiate a different approach for this problem, by defining a simple constraint generation algorithm relying on a cut-based ILP. This method can be seen as an application of \emph{Benders' decomposition}~\cite{Ben62}, where one wants to solve a (generally large) ILP called \emph{master problem} by decomposing it into a smaller (and easier to solve) one, adding new constraints from the master problem when the obtained solution is infeasible (this approach is sometimes known as \emph{row generation}, because new constraints are added throughout the resolution).
We first present different approaches for the addition of new constraints and compare their efficiency on random instances.
We then evaluate the performance of our method by comparing it to the MILP formulation of Dar \etal on randomly generated instances (using the same random generator).

Finally, we present an algorithm for enumerating all optimal solutions of an instance, which we compare to the approach developed by Agarwal \etal.

\paragraph{Organization of the paper.} In the next section, we introduce our constraint generation algorithm. In Section~\ref{sec:experiments}, we recall the random generator of Dar \etal and present the results of the comparison between our constraint generation algorithm and the flow-based MILP formulation. Finally, Section~\ref{sec:enum} is devoted to our enumeration algorithm.



\section{Constraint generation algorithm for \mcishort}
\subsection{Presentation}

Rather than defining a single (M)ILP model whose optimal solutions coincide with optimal solutions of the \mcishort problem,  our approach is a \textit{constraint generation} algorithm which starts with a simple ILP whose optimal solutions do not necessarily correspond to feasible solutions for \mcishort. Then, some constraints are added to the model which is solved again. This process is repeated until we reach a feasible solution.

Let us define more formally our approach. In the sequel, $H = (V, \E)$ will always denote our input hypergraph, and $n$ and $m$ will always denote the number of vertices and hyperedges of $H$, respectively. Recall that $K(H)$ denotes the graph with vertex set $V$ having an edge $uv$ iff $u$ and $v$ belong to a same hyperedge.
Let us first define our starting ILP model. It has one binary variable $x_e$ for every possible edge $e$ of $K(H)$, which takes value $1$ iff the corresponding edge is in the solution. In the following, we will thus make no distinction between solutions of our ILP and graphs with vertex set $V$.

The constraints that will be added are defined by \textit{cuts} $(X_1, X_2, \dots, X_r)$, $r \ge 2$, where $X_i \subseteq V$, $X_i \neq \emptyset$ and $X_i \cap X_j = \emptyset$ for every $i, j \in \{1, \dots, r\}$, $i \neq j$. Given a cut $C := (X_1, \dots, X_r)$, we define its corresponding set of edges $E(C) := \{xy \in E(G), x \in X_i, y \in X_j, i \neq j\}$.
Given a set of cuts $\C$, let $\M(\C)$ be the following ILP:

\begin{center}
\fbox{
\begin{tabular}{l l r r}
Minimize & $\sum\limits_{e \in K(H)} x_e$ & \\ 
subject to: & &\\
 & $\sum\limits_{u, v \in S} x_{uv} \ge |S|-1$ & $\forall S \in \E$ & (1)\\
 & $\sum\limits_{e \in E(C)} x_e \ge r-1$ & $\forall C := (X_1, \dots, X_r) \in \C$ & (2)\\
 & $x_e \in \{0, 1\}$ & $\forall e \in K(H)$ &
\end{tabular}
}
\end{center}

Constraints (1) forces the solution to contain at least $|S|-1$ edges within every hyperedge. Although this constraints is not sufficient to guarantee the connectivity in every hyperedge (for instance, two disjoint cycles also satisfy this constraint), its purpose is mainly to speed-up the resolution.

The purpose of constraints (2) is to forbid $X_1$, $\dots$, $X_r$ to be connected components in the solution: it forces the quotient graph\footnote{The \emph{quotient graph} \wrt $X_1$, $\dots$, $X_r$ has $r$ vertices $v_1$, $\dots$, $v_r$, and an edge $v_iv_j$ whenever there is an edge between a vertex of $X_i$ and a vertex of $X_j$, $i \neq j$.} \wrt $X_1$, $\dots$, $X_r$ to contain at least $r-1$ edges. Notice that if $r=2$, then it forces the solution to have an edge between the two parts $X_1$ and $X_2$.

For a set $S \subseteq V$, define $\B_S := \{(X, S \setminus X): X \subseteq S, X \notin \{\emptyset, S\}\}$ the set of cuts constructed from all non-trivial\footnote{A non-trivial partition of a set $V$ is a partition where each set is different from $\emptyset$ and $V$.} bipartitions of $S$, and $\P_S = \{(X_1, \dots, X_r): r \ge 2, X_i \subseteq S$, $X_i \neq \emptyset$, $ \cup_{i=1}^r X_i = S$ and $X_i \cap X_j = \emptyset$ for all $i, j \in \{1, \dots, r\}, i \neq j\}$ the set of cuts constructed from all non-trivial partitions of $S$. Moreover let $\B_H := \bigcup_{S \in \E} \B_S$ and $\P_H := \bigcup_{S \in \E} \P_S$. We have the following:

\begin{proposition}
Optimal solutions of $\M(\P_H)$ are in one-to-one correspondence with optimal solutions of $\M(\B_H)$ which are themselves in one-to-one correspondence with optimal solutions of the \mcishort instance.
\end{proposition}
\begin{proof}
We have $\B_H \subseteq \P_H$, hence a feasible solution of $\M(\P_H)$ is also a feasible solution of $\M(\B_H)$. A feasible solution of $\M(\B_H)$ is also a feasible solution of \mcishort, since otherwise Constraint (2) would not be satisfied for some bipartition of some hyperedge. Finally a feasible solution of \mcishort is a feasible solution of $\M(\P_H)$, otherwise a hyperedge would not induce a connected subgraph.
\qed
\end{proof}

By the previous proposition, it would be sufficient to solve $\M(\B_H)$ or $\M(\P_H)$. However, we have $|\P_H| = \sum_{S \in \E} 2^{|S|}-1$ and $|\B_H| = \sum_{S \in \E} 2^{|S|-1}-1$, which makes these naive ILPs inefficient from a practical point of view. Fortunately, it turns out that for many instances in practice, only a small number of cuts among $\B_H$ (resp. $\P_H$) is actually needed in order to ensure connectivity in every hyperedge. This idea is the basis of our constraint generation algorithm described below.

\medskip 

\begin{algorithm}[H]
\SetAlgoLined
\KwData{ a hypergraph $H=(V, \E)$}
\KwResult{ a solution $G=(V,E)$ } 
	$\C \leftarrow \C_{init}(H)$ \\
	$G \leftarrow solve(\M(\C))$ \\
	\While{$G$ is not feasible}{
		$\C \leftarrow \C \cup newCuts(G)$ \\
		$G \leftarrow solve(\M(\C))$ \\
	}
 \caption{constraint generation algorithm for \mcishort} 
 \label{algo:main}
\end{algorithm}

\medskip

Our strategy is specified by a set of initial cuts of the input hypergraph $\C_{init}(H)$, and a routine $newCuts(G)$ which takes a non-feasible solution $G$ as input, and outputs a set of cuts. If the $newCuts(.)$ routine always returns cuts from $\B_H$ (resp. $\P_H$) that were not considered before, then the algorithm clearly outputs a feasible optimal solution for the problem, since it only stops when a feasible solution is found and, in the worst case, it ends by solving $\M(\B_H)$ (resp. $\M(\P_H)$). This proves that Algorithm~\ref{algo:main} always terminates and returns an optimal solution for \mcishort, provided that the $newCuts(.)$ routine satisfies the property described above. The choices of the initial set of cuts and this routine are described in the next sub-section.\\

\subsection{Choice of cuts}

The choice of cuts is a crucial feature of our algorithm. The main challenge is to find the policies that will lead to a right balance of the number of added constraints: if too few constraints are added in each iteration, then the number of these iterations will increase, which will then result in a lack of efficiency. On the opposite, if too many constraints are added at the beginning and/or in each iteration, then the size of the ILP will increase too quickly, which will slow down the solver, and then result in a lack of efficiency once again.
Here we present a set of initial set of cuts, and three possible $newCuts(.)$ routines. We then conducted an empirical evaluation of these strategies (using the initial set of cuts or not, followed by one of the three $newCuts(.)$ routine, thus defining six possible strategies).

\paragraph{Initial set of cuts.} For every hyperedge $S \in \E$, and every vertex $v \in S$, the idea is to add the cut $(\{v\}, S \setminus \{v\})$. This set of cuts forbids solutions with isolated vertices in every hyperedge. One could also consider cuts $(X, X \setminus S)$ formed from every subset $X \subseteq S$ of a fixed size $q$. However, for $q=2$ already, we noticed a drop of efficiency, mainly caused by the large number of constraints it creates. Hence, we shall initialize $\C_{init}$ with the cuts formed by singletons only. In the sequel, this initial set of cuts will sometimes be called \emph{singleton cuts}.\\
	
\paragraph{The $newCuts(.)$ routine.} Given a non-feasible solution $G$ of \mcishort, recall that we shall add, for every hyperedge $S$ such that $G[S]$ is disconnected, a set of cuts. Let $S$ be such a hyperedge. Notice that the objective is not to guarantee connectivity in the \textit{very next} iteration of the algorithm, but to constrain the model more and more. Let $S_1$, $\dots$, $S_p$ be the connected components of $G[S]$, with $p \ge 2$. We considered three natural ideas for the set of new cuts corresponding to $S$ in this situation:
	\begin{itemize}
		\item \textbf{Routine 1}: add only one cut $(A, B)$ corresponding to a balanced bipartition of the connected components, that is, $A \cup B = S$, $A \cap B = \emptyset$ and $S_i \subseteq A$ or $S_i \subseteq B$ for every $i \in \{1, \dots, p\}$, and the absolute value of $|A|-|B|$ is as minimum as possible. Since the problem of finding a balanced bipartition of a given set of numbers is an NP-hard problem, the computation of the bipartition was done using a polynomial greedy algorithm which considers connected components in decreasing order \wrt their sizes, and iteratively adds each of them to $A$ (resp. $B$) whenever $|A| < |B|$ (resp. $|A| \ge |B|$). Notice that this algorithm provides a $\frac{7}{6}$-approximation of an optimal bipartition, and runs in $O(p \log p)$ time \cite{Gr69}.
		\item \textbf{Routine 2}: add the cut $(S_i, \cup_{j \neq i} S_j)$, for every $i \in \{1,\dots, p\}$. This idea forbids $S_i$ to be disconnected from the rest of $S$ in the next iteration.
		\item \textbf{Routine 3}: add the cut $(S_1, \dots, S_p)$. Here, we simply forbid $G[S]$ to have the exact same connected components in the next round.
	\end{itemize}
Observe that the first two strategies return cuts from the set $\B_H$ defined previously, while the third one returns a cut which belongs to $\P_H$. In all three cases, the routine returns cuts which were not in the model, hence guaranteeing the optimality and termination of our algorithms, as seen previously.


Combining the above choices, it gives six different strategies:
\begin{itemize}
	\item Strategy 1: initial set of cuts: none ; $newCuts(.)$: Routine 1
	\item Strategy 2: initial set of cuts: none ; $newCuts(.)$: Routine 2
	\item Strategy 3: initial set of cuts: none ; $newCuts(.)$: Routine 3
	\item Strategy 4: initial set of cuts: singleton cuts ; $newCuts(.)$: Routine 1
	\item Strategy 5: initial set of cuts: singleton cuts ; $newCuts(.)$: Routine 2
	\item Strategy 6: initial set of cuts: singleton cuts ; $newCuts(.)$: Routine 3
\end{itemize}

After an empirical evaluation of the above strategies for different kind of instances, we observed a similar behaviour for all of them, with a high deviation for seemingly similar instances.
Nevertheless, we could observe that on average, strategies 4, 5, and 6 were more efficient than strategies 1, 2 and 3, especially for instances with a high number of vertices, which suggests that using a non-empty set of initial set of cuts should always be better. 
The closeness of the results for the three routines can be explained by the fact that in practice (in our random instances, all having less than 25 vertices), the number of connected components of every hyperedge of non-feasible solution is usually small (frequently 2 or 3, and often smaller than 5), which leads to similar ILP models to be solved (for instance, when there are only two connected components, all three routines output exactly the same set of cuts).

Our first empirical results suggest that a more fine-grained comparison should be performed in order to better understand which hypergraph parameters influence the efficiency of our different strategies. This approach could then be used in a more general algorithm which would first analyze the instance to solve, and then choose the right strategy to use. Another option would be to run all strategies in parallel in order to obtain the least running time for every instance.

In the sequel, we decided to effectively use the singleton cuts as initial set of cuts, and to use Routine 1 as $newCuts(.)$ (that is, it corresponds to strategy 4 described above).

\section{Experimental evaluation}\label{sec:experiments}
\subsection{Generation of instances}\label{sec:generator}

Our random generator of instances follows the same rules as in the experiment conducted by Dar \etal \cite{DaFiMaSc18}. A given scenario depends on the following features:
\begin{itemize}
	\item \textbf{Number of vertices $n$} of the hypergraph.
	\item \textbf{Density of the hypergraph $d = \frac{m}{n}$}. As in \cite{DaFiMaSc18}, we used the following values: $d \in \{1, 3, 5\}$.
	\item \textbf{Hyperedge size bounds and distributions.} For this parameter, we used the four types defined by \cite{DaFiMaSc18} plus a new fifth type. For the first four, a size is chosen uniformly at random for each hyperedge among prescribed upper and lower bounds:
	\begin{itemize}
		\item Type 1: sizes of hyperedges between $2$ and $n$
		\item Type 2: sizes of hyperedges between $2$ and $\lceil n/2 \rceil$
		\item Type 3: sizes of hyperedges between $\lceil n/4 \rceil$ and $n$
		\item Type 4: sizes of hyperedges between $\lceil n/4 \rceil$ and $\lceil n/2 \rceil$.
	\end{itemize}
	Then, for each hyperedge, vertices are chosen uniformly at random until the desired size is reached.
	For the fifth type, hyperedges are chosen uniformly at random among all possible hyperedges. To do so, for each hyperedge, each vertex is added with probability $1/2$ until the desired number of distinct hyperedges is reached. Hence, the sizes of hyperedges follow a uniform distribution for the first four types, and a gaussian distribution (centered at $\frac{n}{2}$) for the fifth one.
\end{itemize}

In the following, a \emph{scenario} corresponds to a triple $(n, d, Type)$. In all experiments conducted in this paper, 50 instances were generated for each scenario. Moreover, a time limit of 900 seconds (15 minutes) was set for each instance.

\subsection{Comparison with the flow-based MILP formulation}
In this sub-section, we present the results of the comparison between our constraint generation algorithm and the best state-of-the-art exact algorithm for \mcishort, which is the improved flow-based MILP model of Dar \etal~\cite{DaFiMaSc18}. 
As explained in the introduction, this algorithm is itself an improvement of a previous algorithm of Agarwal \etal~\cite{AgArCaCaCo13}. Although both algorithms rely on a flow-based MILP formulation of the problem, the improvement of Dar \etal can be summarized as follows:
\begin{itemize}
	\item The MILP formulation of Dar \etal contains less variables and constraints, mainly because of a factoring of several linearly-dependent constraints in the previous formulation. They also added some new constraints in order to speed-up the resolution.
	\item The algorithm of Dar \etal also contains several pre-processing rules whose purpose is to reduce the number of vertices and hyperedges of the input instance, and thus reduce the size of the MILP formulation. These reduction rules rely on some observations of the problem, dealing with parts of the instances where the structure of an optimal solution can be inferred in polynomial time (\eg when a set of vertices belong to a same set of hyperedges of a large size). Notice that Dar \etal conducted an experimental evaluation of their reduction rules in \cite{DaFiMaSc18b}.
\end{itemize}

For the sake of completeness, we provide the MILP formlulation of Dar \etal. To this end, let us first introduce some notions and definitions. For every hyperedge $S \in \E$ they choose an arbitrary vertex $r_S \in S$ to be the source of the flow which will ensures connectivity. Hence, they define a complete digraph $A(S)$ with vertex set $S$ and, in addition to a variable $x_e$ for every edge of $K(H)$, their model has also a variable $f_a^S$ for every arc $a$ of $A(S)$. For a vertex $v \in S$, $A^-_S(v)$ (resp. $A^+_S(v)$) denotes the set of arcs of $A(S)$ entering $v$ (resp. leaving $v$). The model is the following:

\begin{center}
\fbox{
\begin{tabular}{l l r r}
Minimize & $\sum\limits_{e \in K(H)} x_e$ & \\ 
subject to: & &\\
 & $\sum\limits_{u, v \in S} x_{uv} \ge |S|-1$ & $\forall S \in \E$ & \\
 & $\sum\limits_{a \in A^-_S(v)} f_a^S - \sum\limits_{a \in A^+_S(v)} f_a^S = -1$ & $\forall S \in \E$, $\forall v \in S \setminus r_S$ & \\
 & $f_{uv}^S + f_{vu}^S \le (|S|-1)x_e$ & $\forall S \in \E$, $\forall u, v \in S$ & \\
 & $f_a^S \ge 0$ & $\forall S \in \E$, $\forall a \in A(S)$ & \\
 & $x_e \in \{0, 1\}$ & $\forall e \in K(H)$ &
\end{tabular}
}
\end{center}

Since our goal was mainly to compare the performance of our constraint generation algorithm to a simple (M)ILP formulation, the reduction rules of Dar \etal were not used for both algorithms.
In the sequel, the algorithm of Dar \etal will be denoted by \ilp, and our constraint generation algorithm by \cga.\\

All experiments were conducted on a computer equipped with an Intel$^{\mbox{\textregistered}}$ Xeon$^{\mbox{\textregistered}}$ E5620  processor (64 bits) at 2.4GHz, 24GB of RAM and a Linux system (Ubuntu version 18.04.1 LTS). The implementation of our constraint generation algorithm (Strategy 4 described above) was written and run in SageMath version 8.2 (release date 05/05/2018).
The algorithm of Dar \etal was written\footnote{We used the implementation of \cite{DaFiMaSc18} provided by their authors.} and run in MATLAB$^{\mbox{\textregistered}}$~Released R2016b.
The MILP solver used in both algorithms was CPLEX$^{\mbox{\textregistered}}$~version 12.8 from IBM$^{\mbox{\textregistered}}$.
All algorithms (including all MILP resolutions) were conducted sequentially, \ie not exploiting multi-threading.
Notice that the measured time of the algorithm of Dar \etal only consists in the resolution of the MILP model (the purpose of the MATLAB$^{\mbox{\textregistered}}$ code is thus only to construct the MILP model from the instance), hence the difference of programming languages does not matter for the comparison.\\

For each scenario $(n, d, Type)$, a set of 50 instances were generated and given to both \ilp and \cga. As said previously, for each instance, a time limit of 900 seconds was set. Tables~\ref{tab:n}, \ref{tab:3n} and \ref{tab:5n} represent the results of the comparison for densities $1$, $3$ and $5$, respectively, where the running time is the average running time of all instances solved within the time limit, and the number in brackets indicates the number of instances (out of 50) effectively solved within this limit in the case this number was different from 50. The tables also show the average number of constraints in the MILP formulation of both algorithms: for \ilp it corresponds directly to the number of constraints of the MILP model, while for \cga it corresponds to the number of constraints it had to add in order to be able to solve the instance (hence, it corresponds to the number of constraints in the last ILP solved).

\begin{table}[h!]
\footnotesize
\centering
\begin{tabular}{|C{0.8cm}|C{1cm}|C{2cm}|C{2cm}|C{2cm}|C{2cm}|}
\hline
$n$ & Type & \ilp (sec) & \ilp (con) & \cga (sec) & \cga (con)\\ 
\hline
14 & 1 & 0.40 & 598.56 & 0.05 & 130.54\\ 
 & 2 & 0.10 & 195.26 & 0.02 & 78.42\\ 
 & 3 & 0.40 & 636.28 & 0.05 & 135.02\\ 
 & 4 & 0.12 & 226.58 & 0.03 & 85.70\\ 
 & 5 & 0.31 & 428.56 & 0.04 & 115.12\\
\hline
16 & 1 & 0.84 & 830.22 & 0.07 & 162.08\\ 
 & 2 & 0.16 & 277.54 & 0.04 & 99.18\\ 
 & 3 & 1.05 & 958.52 & 0.08 & 178.28\\ 
 & 4 & 0.25 & 358.30 & 0.06 & 115.60\\ 
 & 5 & 1.75 & 618.86 & 0.18 & 152.50\\
\hline
18 & 1 & 2.58 & 1163.16 & 0.11 & 201.56\\ 
 & 2 & 0.27 & 372.92 & 0.07 & 123.00\\ 
 & 3 & 8.51 & 1263.40 & 0.19 & 219.42\\ 
 & 4 & 0.40 & 466.40 & 0.09 & 139.08\\ 
 & 5 & 3.72 & 826.74 & 0.37 & 187.18\\
\hline
20 & 1 & 24.17 & 1569.76 & 0.34 & 253.22\\ 
 & 2 & 0.52 & 471.60 & 0.11 & 145.18\\ 
 & 3 & 35.33 & 1799.16 & 0.58 & 282.62\\ 
 & 4 & 1.88 & 672.04 & 0.54 & 182.14\\ 
 & 5 & 35.64 & 1158.72 & 2.97 & 250.80\\
\hline
22 & 1 & 60.53 & 2023.32 & 0.85 & 307.94\\ 
 & 2 & 1.05 & 630.04 & 0.27 & 177.80\\ 
 & 3 & 107.17 [49] & 2386.94 & 1.08 & 342.64\\ 
 & 4 & 6.31 & 838.70 & 1.14 & 216.18\\ 
 & 5 & 137.17 [47] & 1504.98 & 12.02 & 315.64\\
\hline
24 & 1 & 119.20 [49] & 2566.84 & 2.15 & 365.82\\ 
 & 2 & 3.76 & 807.26 & 0.54 & 211.76\\ 
 & 3 & 314.43 [39] & 3147.10 & 8.58 & 443.44\\ 
 & 4 & 42.49 [49] & 1140.20 & 4.29 & 272.38\\ 
 & 5 & 344.12 [28] & 1929.75 & 122.41 [49] & 404.24\\
\hline
26 & 1 & 194.88 [38] & 3342.79 & 28.25 & 448.26\\ 
 & 2 & 19.83 & 1019.42 & 4.28 & 253.16\\ 
 & 3 & 365.14 [33] & 3733.33 & 8.39 & 498.22\\ 
 & 4 & 101.67 [48] & 1338.85 & 16.92 & 318.62\\ 
 & 5 & 606.45 [8] & 2382.62 & 285.70 [31] & 478.58\\
\hline
\end{tabular} 
\caption{Comparison of running times and number of constraints between \ilp and \cga for density $1$. Columns labeled with (sec) (resp. (con) represent the average running time (resp. number of constraints).}
\label{tab:n}
\end{table}

\begin{table}[h!]
\centering
\begin{tabular}{|C{0.8cm}|C{1cm}|C{1.8cm}|C{1.8cm}|C{1.8cm}|C{1.8cm}|}
\hline
$n$ & Type & \ilp (sec) & \ilp (con) & \cga (sec) & \cga (con)\\ 
\hline
12 & 1 & 0.90 & 1056.10 & 0.08 & 275.24\\ 
 & 2 & 0.19 & 399.16 & 0.05 & 181.20\\ 
 & 3 & 0.78 & 1153.72 & 0.08 & 291.50\\ 
 & 4 & 0.30 & 469.44 & 0.06 & 198.74\\ 
 & 5 & 0.99 & 814.18 & 0.13 & 253.84\\
\hline
14 & 1 & 2.40 & 1645.76 & 0.15 & 366.24\\ 
 & 2 & 0.38 & 586.10 & 0.08 & 233.46\\ 
 & 3 & 2.75 & 1707.58 & 0.19 & 377.32\\ 
 & 4 & 0.59 & 665.06 & 0.14 & 251.90\\ 
 & 5 & 6.14 & 1254.08 & 0.69 & 340.64\\
\hline
16 & 1 & 9.67 & 2424.56 & 0.42 & 469.52\\ 
 & 2 & 0.97 & 827.48 & 0.21 & 293.36\\ 
 & 3 & 31.98 & 2773.78 & 1.38 & 518.96\\ 
 & 4 & 8.25 & 1056.66 & 1.60 & 340.16\\ 
 & 5 & 155.75 [49] & 1857.67 & 27.56 & 454.78\\
\hline
18 & 1 & 35.95 & 3269.58 & 1.72 & 578.04\\ 
 & 2 & 3.55 & 1107.86 & 0.66 & 356.14\\ 
 & 3 & 187.60 [49] & 3773.76 & 15.92 & 645.78\\ 
 & 4 & 69.12 & 1411.38 & 12.30 & 417.76\\ 
 & 5 & 393.61 [11] & 2458.09 & 241.37 [33] & 566.21\\
\hline
20 & 1 & 178.20 [44] & 4413.09 & 11.05 & 712.24\\ 
 & 2 & 21.53 & 1454.96 & 4.73 & 432.20\\ 
 & 3 & 418.32 [22] & 5395.00 & 115.07 [48] & 829.40\\ 
 & 4 & 367.51 [16] & 1943.69 & 274.18 [33] & 532.66\\ 
 & 5 & -1.00 [0] & 0.00 & 888.23 [1] & 685.00\\
\hline
22 & 1 & 330.37 [29] & 5896.55 & 76.89 & 872.12\\ 
 & 2 & 105.96 [49] & 1901.41 & 23.52 [49] & 519.82\\ 
 & 3 & 544.25 [2] & 6935.50 & 210.25 [35] & 993.80\\ 
 & 4 & -1.00 [0] & 0.00 & 689.20 [8] & 623.13\\ 
 & 5 & -1.00 [0] & 0.00 & -1.00 [0] & 0.00\\
\hline
\end{tabular} 
\caption{Comparison of running times and number of constraints between \ilp and \cga for density $3$. Columns labeled with (sec) (resp. (con)) represent the average running time (resp. number of constraints).}
\label{tab:3n}
\end{table}

\begin{table}[h!]
\centering
\begin{tabular}{|C{0.8cm}|C{1cm}|C{1.8cm}|C{1.8cm}|C{1.8cm}|C{1.8cm}|}
\hline
$n$ & Type & \ilp (sec) & \ilp (con) & \cga (sec) & \cga (con)\\  
\hline
10 & 1 & 0.44 & 1009.94 & 0.06 & 322.38\\ 
 & 2 & 0.19 & 432.96 & 0.05 & 226.90\\ 
 & 3 & 0.47 & 1023.98 & 0.07 & 324.68\\ 
 & 4 & 0.19 & 438.26 & 0.05 & 228.50\\ 
 & 5 & 0.53 & 816.82 & 0.07 & 301.58\\
\hline
12 & 1 & 1.28 & 1717.92 & 0.14 & 452.32\\ 
 & 2 & 0.36 & 674.92 & 0.07 & 303.98\\ 
 & 3 & 2.37 & 1862.66 & 0.19 & 479.42\\ 
 & 4 & 0.76 & 786.34 & 0.13 & 331.62\\ 
 & 5 & 2.63 & 1346.06 & 0.27 & 420.58\\
\hline
14 & 1 & 6.54 & 2684.80 & 0.28 & 601.18\\ 
 & 2 & 0.82 & 989.18 & 0.17 & 390.82\\ 
 & 3 & 10.00 & 2809.78 & 0.42 & 624.30\\ 
 & 4 & 1.79 & 1112.46 & 0.31  & 419.54\\ 
 & 5 & 57.82 & 2108.38 & 5.90 & 568.06\\
\hline
16 & 1 & 27.61 & 3892.42 & 0.95 & 765.86\\ 
 & 2 & 2.35 & 1375.52 & 0.45 & 485.30\\ 
 & 3 & 146.77 [49] & 4414.90 & 6.91 & 843.04\\ 
 & 4 & 73.50 [46] & 1757.11 & 42.70 & 567.12\\ 
 & 5 & 546.07 [20] & 2976.35 & 176.13 [38] & 728.76\\
\hline
18 & 1 & 91.41 [48] & 5527.35 & 3.31 & 963.82\\ 
 & 2 & 15.35 & 1884.84 & 1.61 & 597.76\\ 
 & 3 & 381.68 [30] & 6082.20 & 55.23 & 1050.14\\ 
 & 4 & 357.08 [36] & 2313.11 & 104.60 [46] & 684.29\\ 
 & 5 & 440.86 [1] & 4244.00 & 283.08 [2] & 911.00\\
\hline
20 & 1 & 178.73 [35] & 7266.37 & 17.41 & 1169.50\\ 
 & 2 & 65.76 & 2430.20 & 6.06 & 708.82\\ 
 & 3 & 588.01 [1] & 9003.00 & 347.20 [21] & 1322.05\\
 & 4 & -1.00 [0] & 0.00 & -1.00 [0] & 0.00\\
 & 5 & -1.00 [0] & 0.00 & -1.00 [0] & 0.00\\
\hline
\end{tabular} 
\caption{Comparison of running times and number of constraints between \ilp and \cga for density $5$. Columns labeled with (sec) (resp. (con)) represent the average running time (resp. number of constraints).}
\label{tab:5n}
\end{table}

As we can see in the results, our approach has a much lower average running time compared to the previous algorithm in every scenario. Indeed, on average (for all instances of all scenarios) \cga has a running time more than 13 times smaller than  \ilp. As we could expect, the newly introduced type 5 of instances is the most difficult for both algorithms, certainly because these instances contain much less small hyperedges than the others. This also explains why type 2 instances are often the easiest to solve for both algorithms.
These results also highlights the fact that our algorithm is able to solve larger instances than previously. When considering types 1, 2, 3, and 4 only:
\begin{itemize}
	\item For $m=n$ and $n=26$ for instance, our algorithm is able to solve 100\% of instances within the time limit, while \ilp can only solve less than 85\% of them. 
	\item For $m=3n$, $n=20$, \cga is able to solve 90\% of instances, while \ilp can only solve 66\%.
	\item For $m=5n$ and $n=18$, \cga is able to solve 98\% of instances while \ilp can only solve 82\% of them.
\end{itemize}

Observe also that our algorithm generate much smaller MILP models. Indeed, firstly the number of variables is always smaller, since our models do not contain any flow variables. Secondly, as we can observe in the results, the number of added constraints is roughly 6 times smaller than in the flow-based MILP model. Despite the fact that for each instance our algorithm needs to call the MILP solver several times, calling it on much smaller MILP models offers a better overall running time.\\

We also generated instances with hyperedges sizes bounded by a (small) constant, in order to see how far we could increase the number of vertices for both algorithms. More precisely, we generated instances with hyperedges of size $7$, and density $d \in \{1, 3\}$ (for density $5$, the maximum number of vertices for which our algorithm was able to solve 100\% of the instances was only 300). 

\begin{table}
\centering
\begin{tabular}{ |C{0.8cm}|C{0.6cm}|C{2.7cm}|C{2.7cm}|C{2.3cm}|C{2.3cm}|}\hline
$n$ & $d$ & \ilp (sec) & \ilp (con) & \cga (sec) & \cga (con)\\ 
\hline
30 & 1 & 0.31 & 407.04 & 0.12 & 169.58\\ 
30 & 3 & 5.14 & 1259.36 & 2.03 & 512.14\\ 
\hline
50 & 1 & 0.47 & 699.32 & 0.25 & 286.58\\ 
50 & 3 & 12.85 & 2066.28 & 2.89 & 856.66\\ 
\hline
100 & 1 & 1.33 & 1396.66 & 0.62 & 572.82\\ 
100 & 3 & 39.35 & 4176.84 & 4.85 & 1752.02\\ 
\hline
200 & 1 & 5.38 & 2809.72 & 0.79 & 1132.54\\ 
200 & 3 & 106.51 [46] & 8269.80 & 4.76 & 3436.34\\ 
\hline
300 & 1 & 17.20 & 4209.62 & 1.40 & 1691.58\\ 
300 & 3 & 148.21 [33] & 12415.18 & 8.04 & 5117.94\\ 
\hline
400 & 1 & 41.62 & 5596.66 & 1.61 & 2239.86\\ 
400 & 3 & 220.36 [37] & 16584.57 & 23.73 & 7033.42\\ 
\hline
500 & 1 & 83.89 & 7002.62 & 2.29 & 2792.82\\ 
500 & 3 & 369.10 [46] & 20739.85 & 94.24 & 8969.68\\ 
\hline
750 & 1 & 296.43 & 10521.40 & 15.07 & 4265.36\\ 
750 & 3 & -1.00 [0] & 0.00 & 266.82 & 13454.44\\ 
\hline
1000 & 1 & 627.295 [4] & 14018.50 & 34.54 & 5645.48\\
1000 & 3 & -1.00 [0] & 0.00 & 666.70 [33] & 17785.06\\ 
\hline
\end{tabular}
\caption{Results for instances with hyperedges of size 7. Columns labeled with (sec) (resp. (con)) represent the average running time (resp. number of constraints).}
\label{tab:large}
\end{table}

The differences of running time is even more significant in this experiment. The algorithm of Dar \etal fails to solve 100\% of the instances within the time limit for 200 vertices already (density $3$). Moreover, for density $1$, there is a huge lack of efficiency between 750 and 1000 vertices for the flow-based MILP algorithm, going from 100\% of instances solved to 8\%.
Overall, we can observe that our approach allows to solve instances or a much larger size than the previous algorithm.

\section{Enumeration algorithm}\label{sec:enum}

In this section, we describe an approach to enumerate all optimal solutions of an instance of \mcishort.
When solving an optimization problem using an MILP formulation in which the solution is represented by 0-1 variables, a natural way to obtain an enumeration algorithm consists in adding new constraints in order to forbid previously found solutions. More formally, if the objective of the MILP is
$$Minimize \sum_{i=1}^n x_i$$
where each $x_i$ is a 0-1 variable, then one can forbid a given solution $S \subseteq \{1, \dots, n\}$ represented by the indices of all variables set to 1 by adding the following constraint:
$$\sum_{i \in S} x_i < |S|$$
Hence, forbidding a set of solutions $\mathcal{A}$ can be done by adding $|\mathcal{A}|$ new constraints to the model.
This idea was used by Agarwal \etal~\cite{AgArCaCaCo13} in order to obtain an algorithm enumerating all optimal solutions of an instance of \mcishort. This strategy, although being easy to implement, becomes much less efficient when the number of solutions of the instances increases, because the size of the MILP model becomes too large for the solver. We propose a new method for the enumeration of solutions, which, in a nutshell, consists in forbidding the solutions ``chunk by chunk''. To this end, we iteratively accumulate optimal solutions by exploring the neighborhood of a solution found (the way we explore this neighborhood will be explained later). Once this exploration is done, we forbid all optimal solutions found at the same time. A pseudo-code of this approach is presented in Algorithm~\ref{algo:enum}.

\begin{algorithm}[H]
\SetAlgoLined
\KwData{ a hypergraph $H=(V, \E)$}
\KwResult{$\mathcal{A}$: the set of all optimal solutions of $H$ } 
	$\mathcal{A} \leftarrow \emptyset$ \\
	$c^* \leftarrow$ cost of an optimal solution of $H$ \\
	\While{there exists a solution $S$ of cost $c^*$ which does not belong to $\mathcal{A}$\label{line:while}}{
		$\mathcal{N} \leftarrow$ neighborhood of $S$ \label{line:neighborhood}\\
		$\mathcal{A} \leftarrow$ $\mathcal{A} \cup \mathcal{N}$ \label{line:next}
	}
	\Return{$\mathcal{A}$}
 \caption{Enumeration algorithm for \mcishort} 
 \label{algo:enum}
\end{algorithm} 

Naturally, we use our constraint generation algorithm described previously in order to find new optimal solutions. Notice that once we have found one optimal solution of cost $c^*$, we shall add a new constraint to our ILP in order to find new solutions of size exactly $c^*$ in the next rounds, which usually speeds up the resolution.
We now describe the way we explore the neighborhood of a solution (which corresponds to Line~\ref{line:neighborhood} of Algorithm~\ref{algo:enum}). This step is done by forbidding an arbitrary edge $e$ of the previously found solution, by simply adding a new constraint to our ILP forcing the corresponding variable $x_e$ to $0$. We thus iteratively accumulate new optimal solutions until the solver returns that the obtained ILP does not admit a solution of the desired cost, which means that the exploration of the neighborhood is done. We then remove the newly added constraints used in this routine for the next loop in Line~\ref{line:while} of Algorithm~\ref{algo:enum}.

\begin{figure}[h]
\includegraphics[scale=0.70]{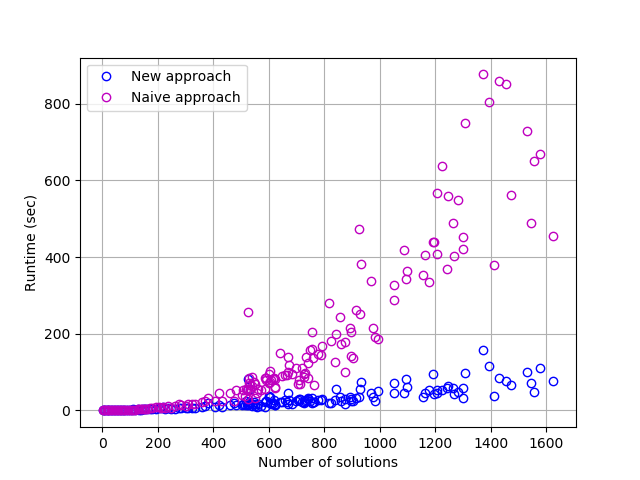}
\label{fig:enum}
\caption{Comparison of running times between the naive enumeration algorithm and our new approach, as a function of the number of solutions of the instances.}
\end{figure}

We evaluated the performance of our approach by comparing its running time to the natural approach of forbidding each new optimal found in the ILP described at the beginning of this section (still using our exact algorithm as a black box for finding new solutions). To this end, we generated a set of 1000 random instances of type 1 with a density of $\frac{m}{n}=2$, and $n=10$. These settings were chosen because they allow the random generation to produce instances of various different structures. In particular, we observed a quite fair distribution of the numbers of solutions, which seemed to be a meaningful parameter for the comparison of the two approaches.
Figure~\ref{fig:enum} presents the result of these experiments. As we can see, our new method offers a great improvement when the number of solutions is high, by reducing by more than 8 times the running time in our generated instances. 
These results suggest that our algorithm has a running time which is linear in the number of solutions in practice.

\section{Conclusion}\label{sec:conclusion}

In this paper we presented and evaluated an exact algorithm for the \mci problem, based on a constraint generation strategy in order to ensure connectivity. Our experiments, conducted on various randomly generated instances, demonstrated that our method outperforms the best previously known exact algorithm for this problem, relying on a flow-based MILP formulation. 
Since connectivity constraints appear very often in practical situations which are usually solved by the means of MILP, our results suggest that a constraint generation strategy can sometimes be much more efficient.
As a further research, it would be interesting to apply this technique to other optimization problems in which connectivity plays an important role.
It should be noted that during the empirical evaluation of the different sub-routines for our algorithm, we noticed high standard deviations in the running times. It would be thus interesting to understand which hypergraph parameters influence the complexity of our strategies. Apart from providing useful information about the problem and our method, this could be used in order to build a more structured benchmark of instances, which could be of great help for the evaluation of future exact algorithms.
Finally, our enumeration algorithms seems to be a promising method which should be tested for other similar problems.
\linebreak

\noindent\textbf{Acknowledgment.} We would like to thank Muhammad Abid Dar, Andreas Fischer, John Martinovic and Guntram Scheithauer for providing us the source code of their algorithm~\cite{DaFiMaSc18}.

%
%
%

%
%

%

\end{document}